\documentclass[12pt]{article}
\usepackage{amsfonts}
\usepackage{amsmath}
\usepackage{amsthm}

\newtheorem{theorem}{Theorem}
\newtheorem{lemma}{Lemma}[section]
\newtheorem{corollary}{Corollary}[section]

\title{Orthogonal decomposition of Lorentz transformations}
\author{jason hanson
        \thanks{\tt jhanson$@$digipen.edu}}
\date{DigiPen Institute of Technology}

\begin{document}
\maketitle
\begin{abstract}
The canonical decomposition of a Lorentz algebra element into a sum of orthogonal simple (decomposable) Lorentz bivectors is defined and discussed.  This decomposition on the Lie algebra level leads to a natural decomposition of a proper orthochronous Lorentz transformation into a product of commuting Lorentz transformations, each of which is the exponential of a simple bivector.  While this later result is known, we present novel formulas that are independent of the form of the Lorentz metric chosen.  As an application of our methods, we obtain an alternative method of deriving the formulas for the exponential and logarithm for Lorentz transformations.
\end{abstract}

\section{Introduction}

It is common to define a three--dimensional rotation by specifying its axis and rotation angle.  However, since the axis of rotation is fixed, we may also define the rotation in terms of the two--plane orthogonal to the axis.  This point of view is useful in higher dimensions, where one no longer has a single rotation axis.  Indeed we may define a rotation in $n$--dimensional space by choosing a two--plane and a rotation angle in that plane; vectors orthogonal to the plane are fixed by the rotation.  The most general rotation, however, is a composition of such planar, or {\em simple}, rotations.  Moreover, the planes of rotation may be taken to be mutually orthogonal to each other, so that the simple factors commute with one another.

On the Lie algebra level, where rotations correspond to antisymmetric matrices, a simple rotation corresponds to a {\em decomposable} matrix; i.e., a matrix of the form $u\wedge v=uv^T-vu^T$, where $u,v$ are (column) vectors.  Geometrically, the vectors $u,v$ span the plane of rotation.  A general rotation corresponds to the sum of decomposable matrices: $A=\sum_iu_i\wedge v_i$, and the mutual orthogonality of the two--planes of rotation corresponds to the mutual annihilation of summands: $(u_i\wedge v_i)(u_j\wedge v_j)=0$ for $i\neq j$.  Note that mutual annihilation trivially implies commutation, so that $\exp(\sum_iu_i\wedge v_i$$)=\prod_i\exp(u_i\wedge v_i)$, and we recover algebraically the commutation of simple factors in a rotation.

In the particular case of four--dimensional Euclidean space, a general rotation is the composition of two simple commuting rotations in orthogonal two--planes.  Correspondingly, a $4\times 4$ antisymmetric matrix can be written as the sum of two mutually annihilating decomposable matrices.  It is possible (although apparently not of sufficient interest to be publishable) to write down explicit formulas for this decomposition --- both on the Lie algebra level and on the level of rotations.  It is natural to ask if such an orthogonal decomposition is possible for Lorentz transformations.  This article answers this question in the affirmative.  If one assumes a Minkowski metric and allows the use of complex numbers, a Wick rotation may be applied to the formulas for four--dimensional Euclidean rotations.  However our discussion will not proceed along this line; instead, we will limit ourselves to algebra over the reals (see below) and develop the formulas from first principles, without any assumptions on the specific form of the Lorentz metric.

It should be noted that the orthogonal decomposition of a Lorentz transformation $\Lambda$ into commuting factors is distinct from the usual polar decomposition $\Lambda=BR$, where $B$ is a boost and $R$ is a rotation.  While $B$ and $R$ are simple in the sense that they are exponentials of decomposable Lie algebra elements: $B=\exp(b)$ and $R=\exp(r)$, they do not necessarily commute, nor do the Lie algebra elements $b,r$ annihilate each other.  In particular, $\exp(b+r)\neq\exp(b)\exp(r)$ in general.

From the point of view of the geometry of Lorentz transformations, there is nothing new presented in this work.  It is known that any proper orthochronous Lorentz transformation is the product of commuting factors which act only on orthogonal two--planes (\cite{Schremp}).  That the corresponding Lorentz algebra element can be written as the sum of mutually annihilating simple elements is a natural consequence.  However, the formulas given here, which actually allow one to compute these decompositions using only basic arithmetic operations, appear to be original.  Moreover they are not metric--specific, and so are applicable to any space--time manifold in a natural way, without the need to use Riemann normal coordinates.

Since it may seem oddly unnecessary to restrict ourselves only to algebra over the reals, we briefly explain our motivation for this.  The author is interested in efficient numerical interpolation of Lorentz transformations, with a view towards a real--time relativistic simulation.  That is, we are given (proper) Lorentz transformations $\Lambda_0$ and $\Lambda_1$, which we think of as giving coordinate transformations, measured at times $t=0$ and $t=1$, between an observer and an object in nonuniform motion.  For any time $t$ with $0\leq t\leq 1$, we would like to find a Lorentz transformation $\Lambda(t)$ that is ``between'' $\Lambda_0$ and $\Lambda_1$.  Analogous to linear interpolation, we may use geodesic interpolation in the manifold of all Lorentz transformations: $\Lambda(t)=\Lambda_0\exp(t\log{\Lambda_0^{-1}\Lambda_1})$.  Thus we need to efficiently compute the logarithm and exponential for Lorentz transformations.  Although many computer languages offer arithmetic over the complex numbers, not all do, and in any case they are typically not supported at the hardware level, so their usage introduces additional computational overhead.

\section{Preliminaries}

\subsection{Generalized trace}\label{sec:traces}

We will make use of the {\bf $k$--th order trace} ${\rm tr}_kM$ of a $n\times n$ matrix $M$, which is defined by the identity
\begin{equation}\label{eq:ktr}
  \det(I+tM)=\sum_{k=0}^n({\rm tr}_kM)t^k.
\end{equation}
Equivalently, the $(n-k)$--th degree term of the characteristic polynomial $\det(\lambda I-M)$ of $M$ has coefficient $(-1)^k{\rm tr}_kM$.  In particular, the zeroth order trace is unity, the first order trace is the usual trace, and the $n$--th order trace is the determinant.  For $k\neq 1$, ${\rm tr}_kM$ is not linear.  However, the identities (i) ${\rm tr}_kM^T={\rm tr}_kM$, (ii) ${\rm tr}_kMN={\rm tr}_kNM$, and (iii) ${\rm tr}_k\alpha M=\alpha^k{\rm tr}_kM$ for any scalar $\alpha$, are always satisfied for any $k$.  These properties follow from equation \eqref{eq:ktr} and the basic properties of determinants.  Of primary interest is the second order trace, for which we have the explicit formula
\begin{equation}\label{eq:2tr}
  {\rm tr}_2M=\tfrac{1}{2}({\rm tr}^2M-{\rm tr}M^2)
\end{equation}
This follows by computing the second derivative of equation \eqref{eq:ktr} and evaluating at $t=0$.  In engineering, ${\rm tr}_2M$ is often denoted by ${\rm II}_M$.

\subsection{Antisymmetry and decomposability}

In preparation for the next section, let us derive a few facts that we will need concerning $4\times 4$ antisymmetric matrices.  Such a matrix can be written in the form
\begin{equation}\label{eq:skew}
   A_{\bf xy}=\begin{pmatrix}
                0        & {\bf x}^T\\
                -{\bf x} & W_{\bf y}
              \end{pmatrix}
  \quad\text{where}\quad
  W_{\bf y}=\begin{pmatrix}
                 0 & -y_3 &  y_2\\
               y_3 &    0 & -y_1\\
              -y_2 &  y_1 &    0
            \end{pmatrix}
\end{equation}
for some ${\bf x},{\bf y}\in{\mathbb R}^3$.  Here $W_{\bf y}$ is chosen so that it has the defining property $W_{\bf y}{\bf v}={\bf y}\times{\bf v}$.

By computing the determinant of \eqref{eq:skew} directly, we obtain the following, which is a special case of the general fact that the determinant of an antisymmetric matrix (in any dimension) is the square of its Pfaffian.

\begin{lemma}\label{lem:skewdet}
$\det{A_{\bf xy}}=({\bf x}\cdot{\bf y})^2$\qed
\end{lemma}

We may identify the {\em wedge (exterior) product} of four--vectors $u,v$ with the matrix $u\wedge v\doteq uv^T-vu^T$.  This matrix is evidently antisymmetric, and if we write $u=(u_0,{\bf u})$, $v=(v_0,{\bf v})$, we find that $u\wedge v=A_{\bf xy}$, where ${\bf x}=u_0{\bf v}-v_0{\bf u}$ and ${\bf y}={\bf v}\times{\bf u}$.  In particular, lemma \ref{lem:skewdet} implies that $\det{u\wedge v}=0$.  This is again a special case of a general fact: the determinant of a wedge is zero in dimensions greater than two, since we can always find a vector orthogonal to the vectors $u,v$.  In dimension four, the converse is also true.

\begin{lemma}\label{lem:wedge}
$A_{\bf xy}$ is a wedge--product if and only if $\det{A_{\bf xy}}=0$.
\end{lemma}

\begin{proof}
For the sufficiency, first assume that ${\bf y}\neq{\bf 0}$.  Choose vectors ${\bf u},{\bf v}$ such that ${\bf v}\times{\bf u}={\bf y}$.  Since ${\bf x}\cdot{\bf y}=0$ (lemma \ref{lem:skewdet}), we may write ${\bf x}=u_0{\bf v}-v_0{\bf u}$ for some scalars $u_0,v_0$.  If ${\bf y}={\bf 0}$, then take $u=(1,{\bf 0})$ and $v=(0,{\bf x})$.
\end{proof}

In general, the vector space of two--forms in $n$ dimensions can be identified with the space of $n\times n$ antisymmetric matrices.  In dimensions $n\geq 4$, not every two--form is {\em decomposable;} i.e., the wedge of two one--forms.  Lemma \ref{lem:wedge} gives a simple test for decomposability in dimension $n=4$.

\section{Lorentz algebra element decomposition}

Let $g$ be a Lorentz inner product on ${\mathbb R}^4$; i.e., $g$ is symmetric, nondegenerate, and $\det{g}<0$.  Recall that the {\bf Lorentz algebra} $so(g)$ is the collection of all linear transformations $L$ on ${\mathbb R}^4$ for which $L^Tg+gL=0$; or equivalently, $L^T=-gLg^{-1}$.  We will refer to an element of $so(g)$ as a {\bf bivector}.

\subsection{Algebra invariants and properties}

\begin{lemma}\label{lem:antisym}
$Ag\in so(g)$ if and only if $A$ is antisymmetric.
\end{lemma}

\begin{proof}
Since $g$ is symmetric, $(Ag)^Tg+g(Ag)=g(A^T+A)g$, which is zero if and only if $A^T=-A$.
\end{proof}

Thus any Lorentz bivector $L$ can be written in the form $L=Ag$ for some antisymmetric matrix $A$.  Although this identification gives us a vector space isomorphism between $so(g)$ and the space of all antisymmetric $4\times 4$ matrices, the identification is not compatible with their respective Lie algebra structures; i.e., does not give a Lie algebra isomorphism.

\begin{lemma}\label{lem:Ltr13}
For any $L\in so(g)$, ${\rm tr}L=0={\rm tr}_3L$.
\end{lemma}

\begin{proof}
${\rm tr}_kL={\rm tr}_kL^T={\rm tr}_k(-gLg^{-1})={\rm tr}_k(-L)=(-1)^k{\rm tr}_kL$.
\end{proof}

\begin{lemma}\label{lem:detL}
$\det{L}\leq 0$ for all $L\in so(g)$.
\end{lemma}

\begin{proof}
Write $L=Ag$.  Then $\det{L}=\det{A}\,\det{g}\leq 0$, by lemma \ref{lem:skewdet}.
\end{proof}

By the comments in section \ref{sec:traces}, the characteristic equation for a Lorentz bivector $L$ is
\begin{equation}\label{eq:charL}
  \lambda^4+({\rm tr}_2L)\lambda^2+\det{L}=0.
\end{equation}
The squared--roots (that is, the roots of $x^2+({\rm tr}_2L)x+\det{L}=0$) are
\begin{equation}\label{eq:eigenL}
  \mu_\pm=\tfrac{1}{2}\bigl(-{\rm tr}_2L
            \pm\sqrt{{\rm tr}_2^2L-4\det{L}}\bigr)
\end{equation}
Since $\det{L}\leq 0$, $\mu_\pm$ are real numbers.  Note that they are solutions to the equations $\mu_++\mu_-=-{\rm tr}_2L$ and $\mu_+\mu_-=\det{L}$, and that $\mu_+\geq 0$ and $\mu_-\leq 0$.  Moreover, $\det{L}\neq 0$ if and only if $\mu_+>0$ and $\mu_-<0$.

\begin{lemma}\label{lem:4}
$L^4+({\rm tr}_2L)L^2+(\det{L})I=0$.
\end{lemma}

\begin{proof}
This follows from equation \eqref{eq:charL} and the Cayley--Hamilton theorem.
\end{proof}

The determinant and second order trace of a Lorentz bivector may be computed from the traces of its powers.

\begin{lemma}\label{lem:pow}
${\rm tr}_2L=-\tfrac{1}{2}{\rm tr}L^2$ and $\det{L}=\tfrac{1}{8}({\rm tr}^2L^2-2\,{\rm tr}L^4)$.
\end{lemma}

\begin{proof}
Lemma \ref{lem:Ltr13} and equation \eqref{eq:2tr} imply the first identity.  For the second, lemma \ref{lem:4} implies ${\rm tr}L^4+{\rm tr}_2L\,{\rm tr}L^2+4\det{L}=0$.
\end{proof}

\subsection{Simple bivectors}

Given two four--vectors $u,v$, we may construct the {\bf simple Lorentz bivector} $u\wedge^g v$, which is the $4\times 4$ matrix
\begin{equation}\label{eq:bivector}
  u\wedge^g v\doteq uv^Tg-vu^Tg
\end{equation}
In covariant coordinate notation, $(u\wedge^gv)_\alpha^\beta=u^\beta v_\alpha-v^\beta u_\alpha$.  Observe that $(u\wedge^g v)=(u\wedge v)g$, whence lemma \ref{lem:antisym} implies that $u\wedge^g v\in so(g)$.  

\begin{lemma}\label{lem:simple}
$L\in so(g)$ is simple if and only if $\det{L}=0$.
\end{lemma}

\begin{proof}
Write $L=Ag$, with $A$ antisymmetric.  Since $\det{L}=\det{A}\det{g}$, the lemma follows from lemma \ref{lem:wedge}.
\end{proof}

\begin{lemma}\label{lem:3}
$L\in so(g)$ is simple if and only if $L^3+({\rm tr}_2L)L=0$.  Moreover if $L=u\wedge^gv$, then ${\rm tr}_2L=(u^Tgu)(v^Tgv)-(u^Tgv)^2$.
\end{lemma}

\begin{proof}
For $L=u\wedge^gv$, one computes directly from \eqref{eq:bivector} that ($\ast$) ${\rm tr}L^2=2\gamma$ and $L^3=\gamma L$, where $\gamma$ is the scalar $(u^Tgv)^2-(u^Tgu)(v^Tgv)$.  Lemma \ref{lem:pow} and ($\ast$) imply that ${\rm tr}_2L=-\gamma$.  For the converse, $L^3+({\rm tr}_2L)L=0$ implies that $L^4+({\rm tr}_2L)L^2=0$.  Lemma \ref{lem:pow} then yields $\det{L}=-\tfrac{1}{4}({\rm tr}L^4+{\rm tr}_2L\,{\rm tr}L^2)=0$, and lemma \ref{lem:simple} applies.
\end{proof}

The four--vectors $u,v$ are parallel if and only if $u\wedge^gv=0$.  However, even if $u,v$ are linearly independent (so that $u\wedge^gv\neq 0$), it is possible that the metric is no longer nondegenerate when restricted to the plane spanned by $u,v$.  We will call such a two--plane {\bf degenerate}.  The next lemma (combined with the previous lemma) shows that degenerate two--planes are synonymous with {\em null two--planes (two--flats)} in the sense of \cite{Synge}.

\begin{lemma}\label{lem:degen}
$u,v$ span a nondegenerate two--plane if and only if ${\rm tr}_2u\wedge^gv\neq 0$.
\end{lemma}

\begin{proof}
For $w=\alpha u+\beta v$, we have $w^Tgu=\alpha(u^Tgu)+\beta(v^Tgu)$ and $w^Tgv=\alpha(u^Tgv)+\beta(v^Tgv)$.  This linear system has a unique solution if and only if its determinant $(u^Tgu)(v^Tgv)-(u^Tgv)^2={\rm tr}_2L$ is nonzero.
\end{proof}

The following lemma implies that almost all simple Lorentz bivectors are determined up to constant multiple by their squares.

\begin{lemma}\label{lem:2plane}
Suppose $L=u\wedge^gv$ is such that ${\rm tr}_2L\neq 0$. Then $P_L\doteq-L^2/{\rm tr}_2L$ is orthogonal projection (with respect to $g$) onto the nondegenerate two--plane spanned by $u,v$; that is, $P_L^2=P_L$, ${\rm tr}P_L=2$, and $(gP_L)^T=gP_L$.  Conversely, if $P$ is orthogonal projection onto a two--plane, the two--plane is necessarily nondegenerate, and we have $P=P_L$, where $L=u\wedge^gv$ and $u,v$ are any linearly independent four--vectors in the image of $P$.
\end{lemma}

\begin{proof}
The first statement follows from lemmas \ref{lem:3}, \ref{lem:pow}, and \ref{lem:degen}.  For the second statement, let ${\mathcal P}$ denote the two--plane forming the image of $P$.  Suppose that $w\in{\mathcal P}$ is such that $w^Tgy=0$ for all $y\in{\mathcal P}$.  Thus for any four--vector $x$, $w^Tgx=(Pw)^Tgx=w^TgPx=0$; hence $w=0$, since $g$ is nondegenerate.  Thus ${\mathcal P}$ is nondegenerate.  The remainder of the lemma follows from the uniqueness of orthogonal projection.  Indeed, given a four--vector $x$, for any $y\in{\mathcal P}$ we have $(Px)^Tgy=x^TgPy=x^Tgy$.  Since $g$ is nondegenerate on $\mathcal P$, the value of $Px$ is thus uniquely determined by $x$ and $g$.
\end{proof}

\subsection{Orthogonal sum of simple bivectors}

Although an element $L\in so(g)$ will not be a simple bivector in general, we may write it as a sum of simple bivectors.  Such a sum is not unique; however, we will do so in a canonical way using the operators
\begin{equation}\label{eq:proj}
  P_\pm\doteq\pm\frac{L^2-\mu_\mp I}{\mu_+-\mu_-}
\end{equation}
Here $\mu_\pm$ are defined as in equation \eqref{eq:eigenL}.  For $P_\pm$ to be well--defined, $L$ cannot be a simple bivector; that is, $\det{L}\neq 0$, since by the comments following equation \eqref{eq:eigenL}, this guarantees that $\mu_+-\mu_-\neq 0$.

\begin{theorem}\label{thm:decomp}
If $L\in so(g)$ is such that $\det{L}\neq 0$, then (i) $P_++P_-=I$, (ii) $P_+P_-=0=P_-P_+$, (iii) $P_\pm^2=P_\pm$, and (iv) $P_\pm L=LP_\pm$.  Moreover, $P_\pm L$ are simple Lorentz bivectors with $L=P_+L+P_-L$, $(P_+L)(P_-L)=0=(P_-L)(P_+L)$, and ${\rm tr}_2P_\pm L=-\mu_\pm$.
\end{theorem}

\begin{proof}
Set $N\doteq\mu_+-\mu_-$.  Then $L^4+({\rm tr}_2L)L^2+(\det{L})I=(L^2-\mu_+I)(L^2-\mu_-I)=-N^2P_+P_-$.  Thus $P_+P_-=0$ by lemma \ref{lem:4}.  Properties (i), (iii), (iv) also follow from the definitions of $P_\pm$, coupled with lemma \ref{lem:4} and the remarks after equation \eqref{eq:eigenL}.  The remaining two algebraic statements follow from (i), (ii), and (iv), so it suffices to establish that $P_\pm L$ are simple Lorentz bivectors with the second order trace as stated.  First, note that $(L^3)^T=(L^T)^3=(-gLg^{-1})^3=-gL^3g^{-1}$.  It then follows from equation \eqref{eq:proj} that $(P_\pm L)^T=-g(P_\pm L)g^{-1}$; hence $P_\pm L\in so(g)$.  Second, $(P_\pm L)^3=P_\pm L^3=\pm(L^5-\mu_\mp L^3)/N$.  Using lemma \ref{lem:4}, we have $L^5=-({\rm tr}_2L)L^3-(\det{L})L$.  Therefore, $(P_\pm L)^3=\pm\mu_\pm(L^3-\mu_\mp L)/N=\mu_\pm P_\pm L$.  From lemma \ref{lem:3}, $P_\pm L$ is a simple bivector with ${\rm tr}_2P_\pm L=-\mu_\pm$.
\end{proof}

\begin{corollary}\label{cor:bivector}
Any nonsimple Lorentz bivector $L$ may be written as the sum of two simple Lorentz bivectors: $L=L_++L_-$ with $L_+L_-=0=L_-L_+$, where
$$L_\pm=\pm\frac{L^3-\mu_\mp L}{\mu_+-\mu_-}$$
and ${\rm tr}_2L_\pm=-\mu_\pm$, with $\mu_\pm$ are defined as in equation \eqref{eq:eigenL}.\qed
\end{corollary}

We shall refer to the decomposition $L=L_++L_-$ in the corollary as the {\bf orthogonal decomposition} of $L$.  We remark that since $\mu_+>0$ and $\mu_-<0$ for a nonsimple Lorentz bivector, lemma \ref{lem:3} implies that the two--plane associated to $L_+$ (as in lemma \ref{lem:2plane}) is {\em time--like} (intersects the null--cone) while that of $L_-$ is {\em space--like} (does not intersect the null--cone) in Synge's classification of two--planes given in \cite{Synge}.

\subsection{Special case: Minkowski metric}

In the case when $g=\eta\doteq{\rm diag}(-1,1,1,1)$ is the Minkowski metric, we may make use of the explicit parametrization of antisymmetric matrices given in equation \eqref{eq:skew}.  That is, any element of $so(\eta)$ can be written in the form
\begin{equation}\label{eq:minkowski}
  L_{\bf xy}
  =\begin{pmatrix}
     0       & {\bf x}^T\\
     {\bf x} & W_{\bf y}
   \end{pmatrix}
\end{equation}
for some ${\bf x}, {\bf y}\in{\mathbb R}^3$; i.e., $L_{\bf xy}=A_{\bf xy}\eta$.  Note that $\det{L_{\bf xy}}=-({\bf x}\cdot{\bf y})^2$, so that $L_{\bf xy}$ is simple if and only if ${\bf x}\cdot{\bf y}=0$.  One computes, using lemma \ref{lem:pow}, that ${\rm tr}_2L_{\bf xy}=|{\bf y}|^2-|{\bf x}|^2$.

\begin{theorem}
If ${\bf x}\cdot{\bf y}\neq 0$, then $L_{\bf xy}=L_{{\bf a}_+{\bf b}_+}+L_{{\bf a}_-{\bf b}_-}$ is the orthogonal decomposition of $L_{\bf xy}$: $L_{{\bf a}_+{\bf b}_+}L_{{\bf a}_-{\bf b}_-}=0=L_{{\bf a}_-{\bf b}_-}L_{{\bf a}_+{\bf b}_+}$, where
$${\bf a}_\pm=\pm\frac{\mu_\pm{\bf x}+({\bf x}\cdot{\bf y}){\bf y}}
                      {\mu_+-\mu_-}
  \quad\text{and}\quad
  {\bf b}_\pm=\pm\frac{\mu_\pm{\bf y}-({\bf x}\cdot{\bf y}){\bf x}}
                      {\mu_+-\mu_-}
$$
and $\mu_\pm=\tfrac{1}{2}\bigl(|{\bf x}|^2-|{\bf y}|^2\pm\sqrt{(|{\bf x}|^2-|{\bf y}|^2)^2+4({\bf x}\cdot{\bf y})^2}\bigr)$.
\end{theorem}

\begin{proof}
Using equation \eqref{eq:minkowski}, compute $L_{\bf xy}^3=L_{{\bf x}'{\bf y}'}$, where ${\bf x}'=(|{\bf x}|^2-|{\bf y}|^2){\bf x}+({\bf x}\cdot{\bf y}){\bf y}$, and ${\bf y}'=-({\bf x}\cdot{\bf y}){\bf x}+(|{\bf x}|^2-|{\bf y}|^2){\bf y}$.  Now use corollary \ref{cor:bivector}.
\end{proof}

One may show that ${\bf a}_+\cdot{\bf a}_-=0$, ${\bf b}_+\cdot{\bf b}_-=0$, and ${\bf a}_\pm\times{\bf b}_\mp=0$.  Necessarily, ${\bf a}_\pm\cdot{\bf b}_\pm=0$, since $L_{{\bf a}_+{\bf b}_+}$ and $L_{{\bf a}_-{\bf b}_-}$ are simple.

\subsubsection{Relation with the Hodge dual}

We may also interpret orthogonal decomposition in terms of Hodge duals.  Write the Lorentz bivector $L$ as the two--form $L=\tfrac{1}{2}L_\beta^\alpha e^\beta\wedge e_\alpha$, where $e_\alpha$ ($\alpha=0,1,2,3$) is a basis of ${\mathbb R}^4$ and $e^\beta=\eta^{\alpha\beta}e_\alpha$ is the dual basis (the summation convention is used).  Observe that we may do this as $L\eta^{-1}=L\eta$ is antisymmetric.  The Hodge dual (with respect to $\eta$) of $A_{\mathbf x\mathbf y}$ is $A_{(-{\mathbf y}){\mathbf x}}$, so that $\ast L_{{\mathbf x}{\mathbf y}}=L_{(-{\mathbf y}){\mathbf x}}$.  Moreover, one computes that $({\mathbf x}\cdot{\mathbf y}){\mathbf a}_\pm=\mu_\pm{\mathbf b}_\mp$ and $({\mathbf x}\cdot{\mathbf y}){\mathbf b}_\pm=-\mu_\pm{\mathbf a}_\mp$.  It then follows that
$$*L_{{\mathbf a}_+{\mathbf b}_+}
  =\frac{\mu_+}{{\mathbf x}\cdot{\mathbf y}}\,
     L_{{\mathbf a}_-{\mathbf b}_-}
  \quad\text{and}\quad
  *L_{{\mathbf a}_-{\mathbf b}_-}
  =-\frac{{\mathbf x}\cdot{\mathbf y}}{\mu_+}\,
     L_{{\mathbf a}_+{\mathbf b}_+}.
$$

\subsubsection{Relation with the spin group}\label{sssec:spin}

In addition, we describe the orthogonal decomposition in terms of the double covering group homomorphism $\Psi:{\it SL_2}{\mathbb C}\rightarrow{\it SO^+}(\eta)$.  Since we are using the mathematicians choice of metric signature, we use $\rho(u)\doteq\bigl(\begin{smallmatrix}u_1+iu_2 & u_0+u_3\\ u_0-u_3 & u_1-iu_2\end{smallmatrix}\bigr)$ as our linear embedding of ${\mathbb R}^4$ into the algebra of $2\times 2$ complex matrices, along with the involution $\bigl(\begin{smallmatrix}A & B\\ C & D\end{smallmatrix}\bigr)^\star\doteq\bigl(\begin{smallmatrix}\bar{D} & \bar{B}\\ \bar{C} & \bar{A}\end{smallmatrix}\bigr)$, so that $\rho(u)^\star=\rho(u)$ and $\det\rho(u)=u^T\eta u$.  The map $\Psi$ sends the $2\times 2$ complex matrix $M$ with unit determinant to the Lorentz transformation $u\mapsto\rho^{-1}(M\rho(u)M^\star)$.

On the Lie algebra level, the map $\psi:{\it sl_2}{\mathbb C}\rightarrow{\it so}(\eta)$ corresponding to $\Psi$ sends the traceless complex matrix $m=\bigl(\begin{smallmatrix} a & b\\ c & -a\end{smallmatrix}\bigr)$ to the Lorentz bivector $u\mapsto\rho^{-1}(m\rho(u)+\rho(u)m^\star)$.  One computes that $\psi(m)=L_{{\mathbf x}{\mathbf y}}$, where ${\mathbf x}=\bigl({\it Re}(b+c),{\it Im}(b-c),2{\it Re}(a)\bigr)$ and ${\mathbf y}=\bigl({\it Im}(b+c),-{\it Re}(b-c),2{\it Im}(a)\bigr)$.  Thus ${\mathbf x}\cdot{\mathbf y}=2{\it Im}(a^2+bc)$ and $|{\mathbf y}|^2-|{\mathbf x}|^2=-4{\it Re}(a^2+bc)$.  In particular, $\psi(m)$ is simple if and only if $a^2+bc$ is real.  While it is possible to work out $\psi^{-1}(L_{{\mathbf a}_+{\mathbf b}_+})$, the resulting expression is not particularly nice.  However, one computes that $\psi(im)=\ast L_{{\mathbf x}{\mathbf y}}$, so that by the previous discussion, if $\psi(m_+)=L_{{\mathbf a}_+{\mathbf b}_+}$, then $\psi(i\theta m_+)=L_{{\mathbf a}_-{\mathbf b}_-}$ for some real $\theta$.

To get back to the Lie group level, we only need to exponentiate: if $\psi(m)=L$, then $\Psi(\exp(m))=\exp(L)$.  We will return to this briefly at the end of section \ref{ssec:factordecomp}.

\section{Exponential on $so(g)$}

Recall that the exponential of a square matrix $M$ is defined as $\exp(M)=\sum_{n=0}^\infty\frac{1}{n!}M^n$.  The series always converges, and the exponential of a Lorentz bivector lies in the identity component the Lie group of all Lorentz transformations.

\begin{theorem}\label{thm:expsimple}
If $L$ is a simple Lorentz bivector, then $\exp(L)=I+d_1L+d_2L^2$, where
\begin{align*}
  &d_1=\frac{\sinh\sqrt{-{\rm tr}_2L}}{\sqrt{-{\rm tr}_2L}}
  &\text{and} &
  &d_2=\frac{1-\cosh\sqrt{-{\rm tr}_2L}}{{\rm tr}_2L} &
  &\text{if ${\rm tr}_2L<0$,}\\
  &d_1=\frac{\sin\sqrt{{\rm tr}_2L}}{\sqrt{{\rm tr}_2L}}
  &\text{and} &
  &d_2=\frac{1-\cos\sqrt{{\rm tr}_2L}}{{\rm tr}_2L} &
  &\text{if ${\rm tr}_2L>0$,}
\end{align*}
and if ${\rm tr}_2L=0$, then $d_1=1$ and $d_2=\tfrac{1}{2}$.
\end{theorem}

\begin{proof}
Lemma \ref{lem:3} implies that $L^{2k+1}=(-1)^k({\rm tr}_2L)^kL$ for all $k\neq 0$, which is used to sum the exponential series.
\end{proof}

This formula appears in \cite{Geyer} with $-{\rm tr}_2L$ replaced by the symbol $\alpha$; however in that reference, $\alpha$ is not identified with a matrix invariant, the coefficient of $L^2$ is incorrect, and the limiting case when ${\rm tr}_2L=0$ is not handled explicitly.

For nonsimple Lorentz bivectors, one may sum the exponential series using the identity $L^k=\mu_+^kL_++\mu_-^kL_-$ for all $k\geq 0$, which one verifies by induction in conjunction with lemma \ref{lem:4}.  The resulting formula, which is cubic in $L$ appears in \cite{Geyer} and in \cite{Coll}, although the formula in the latter takes a different form than that in the former; both references use somewhat different methods than ours, and in particular, make use of algebra over the complex numbers.

Alternatively, we may make use of corollary \ref{cor:bivector} to obtain the exponential of a nonsimple Lorentz algebra element.  Since $L_+,L_-$ trivially commute, we have $\exp(L)=\exp(L_+)\exp(L_-)$, and each factor can be computed from theorem \ref{thm:expsimple} to obtain the following.

\begin{theorem}\label{thm:exp}
If $L\in so(g)$ is nonsimple and $L=L_++L_-$ is the orthogonal decomposition of $L$ as in corollary \ref{cor:bivector}, then
$$\exp(L)=I+d_1^+L_++d_1^-L_-+d_2^+L_+^2+d_2^-L_-^2$$
\begin{align*}
  d_1^+ &=\sinh\sqrt{\mu_+}/\sqrt{\mu_+}
   & d_1^- &=\sin\sqrt{-\mu_-}/\sqrt{-\mu_-}\\
  d_2^+ &=(\cosh\sqrt{\mu_+}-1)/\mu_+
   & d_2^- &=(\cos\sqrt{-\mu_-}-1)/\mu_-\qed
\end{align*}
\end{theorem}

Note that if we write $L_\pm$ in terms of $L$ (according to corollary \ref{cor:bivector}), we obtain a polynomial of degree six in $L$.  However, the degree can be reduced to three by lemma \ref{lem:4}; the resulting expression is necessarily the same as that obtained in \cite{Coll} and \cite{Geyer}.

\section{Lorentz transformation factorization}

The {\bf Lorentz group} $O(g)$ is set of all linear transformations $\Lambda$ on ${\mathbb R}^4$ that preserve the metric: $\Lambda^Tg\Lambda=g$.  Necessarily $\det\Lambda=\pm 1$, and $\Lambda$ is said to be {\bf proper} if $\det\Lambda=1$.  It is well--known that the group of all proper Lorentz transformations has two connected components; the Lorentz transformations within the identity component are said to be {\bf orthochronous}.  In the case when $g=\eta$ is the Minkowski metric, the orthochronous transformations are characterized by the component $\Lambda_0^0$ being positive.  The set of all proper orthochronous Lorentz group elements will be denoted by $SO^+(g)$.

{\em A note on proofs.}  In the remainder, all proofs will use the following abbreviations without warning: $\tau_k\doteq{\rm tr}_k\Lambda$, $l_k\doteq{\rm tr}_kL$, $x\doteq\sqrt{-l_2}$ when $l_2<0$, $y\doteq\sqrt{l_2}$ when $l_2>0$, $c_+\doteq\cosh{x}$, $c_-\doteq\cos{y}$, $s_+\doteq\sinh{x}/x$, and $s_-\doteq\sin{y}/y$.  Note that since $\Lambda^{-1}=g^{-1}\Lambda^Tg$, we have ${\rm tr}_k\Lambda^{-1}=\tau_k={\rm tr}_k\Lambda$.  We remind the reader that ${\rm tr}_2\Lambda=\tfrac{1}{2}({\rm tr}^2\Lambda-{\rm tr}\Lambda^2)$, from equation \eqref{eq:2tr}.

\subsection{Some Lorentz transformation relations}

\begin{lemma}\label{lem:trace13}
For all $\Lambda\in SO^+(g)$, ${\rm tr}_3\Lambda={\rm tr}\Lambda$.  Consequently, $\Lambda$ satisfies the relation $\Lambda^4-({\rm tr}\Lambda)\Lambda^3+({\rm tr}_2\Lambda)\Lambda^2-({\rm tr}\Lambda)\Lambda+I=0$.
\end{lemma}

\begin{proof}
The characteristic polynomials of $\Lambda$ and $\Lambda^{-1}$ are the same, namely $\lambda^4-\tau_1\lambda^3+\tau_2\lambda^2-\tau_3\lambda+1$.  By the Cayley--Hamilton theorem, ($\ast$) $\Lambda^4-\tau_1\Lambda^3+\tau_2\Lambda^2-\tau_3\Lambda+I=0$ and ($\ast\ast$) $\Lambda^{-4}-\tau_1\Lambda^{-3}+\tau_2\Lambda^{-2}-\tau_3\Lambda^{-1}+I=0$.  Multiplying ($\ast\ast$) by $\Lambda^4$ and comparing with ($\ast$), we get $\tau_3=\tau_1$.
\end{proof}

\begin{lemma}\label{lem:spow}
If $\Lambda\in SO^+(g)$, then $\Lambda^k+\Lambda^{-k}=A_kI+B_k(\Lambda+\Lambda^{-1})$, where $A_2=-{\rm tr}_2\Lambda$, $B_2={\rm tr}\Lambda$, $A_3=(2-{\rm tr}_2\Lambda){\rm tr}\,\Lambda$, $B_3={\rm tr}^2\Lambda-{\rm tr}_2\Lambda-1$, $A_4=(2-{\rm tr}_2\Lambda)\,{\rm tr}^2\Lambda+{\rm tr}_2^2\Lambda-2$, and $B_4=({\rm tr}^2\Lambda-2\,{\rm tr}_2\Lambda)\,{\rm tr}\Lambda$
\end{lemma}

\begin{proof}
From lemma \ref{lem:trace13}, we have $\Lambda^2-\tau_1\Lambda+\tau_2I-\tau_1\Lambda^{-1}+\Lambda^{-2}=0$; so that $\Lambda^2+\Lambda^{-2}=-\tau_2I+\tau_1(\Lambda+\Lambda^{-1})$.  Moreover, $\Lambda^3-\tau_1\Lambda^2+\tau_2\Lambda-\tau_1I+\Lambda^{-1}=0$ and $\Lambda^{-3}-\tau_1\Lambda^{-2}+\tau_2\Lambda^{-1}-\tau_1I+\Lambda=0$; thus $\Lambda^3+\Lambda^{-3}=\tau_1(\Lambda^2+\Lambda^{-2})-(\tau_2+1)(\Lambda+\Lambda^{-1})+2\tau_1I=\tau_1(2-\tau_2)I+(\tau_1B_2-\tau_2-1)(\Lambda+\Lambda^{-1})$.  Similarly, we have $\Lambda^4+\Lambda^{-4}=\tau_1(\Lambda^3+\Lambda^{-3})-\tau_2(\Lambda^2+\Lambda^{-2})+\tau_1(\Lambda+\Lambda^{-1})-2I=(\tau_1A_3-\tau_2A_2-2)I+(\tau_1B_3-\tau_2B_2+\tau_1)(\Lambda+\Lambda^{-1})$.
\end{proof}

\subsection{Simple Lorentz transformations}

We will say that a Lorentz transformation is {\bf simple} if it is the exponential of a simple Lorentz bivector.  Our first goal is to give an algebraic criterion for simplicity.  We will make use of the fact that the orthochronous Lorentz group is exponential (see \cite{Nishikawa}); that is, $\exp:so(g)\rightarrow SO^+(g)$ is surjective.

\begin{lemma}\label{lem:lorentztraces}
For nonsimple $\Lambda\in SO^+(g)$, there is a nonsimple Lorentz bivector $L$ with orthogonal decomposition $L=L_++L_-$ such that $\Lambda=\exp(L)$ and ${\rm tr}_2L_\pm=-\mu_\pm$, where $\sqrt{\mu_+}=\cosh^{-1}\tfrac{1}{4}({\rm tr}\Lambda+\sqrt\Delta)$, $\sqrt{-\mu_-}=\cos^{-1}\tfrac{1}{4}({\rm tr}\Lambda-\sqrt\Delta)$, and $\Delta={\rm tr}^2\Lambda-4{\rm tr}_2\Lambda+8$.  Moreover, we have ${\rm tr}\Lambda=2(c_++c_-)$ and ${\rm tr}_2\Lambda=4c_+c_-+2$, where $c_+\doteq\cosh\sqrt{\mu_+}$, $c_-\doteq\cos\sqrt{-\mu_-}$.
\end{lemma}

\begin{proof}
Since $\Lambda$ is not simple, $\Lambda=\exp(L)$ for some nonsimple Lorentz bivector $L$, and theorem \ref{thm:exp} applies.  Moreover, using lemma \ref{lem:3} one computes
\begin{equation}\label{eq:2}
  \Lambda^2
  =I+2d_1^+(1+\mu_+d_2^+)L_++2d_1^-(1+\mu_-d_2^-)L_-
   +A_+L_+^2+A_-L_-^2
\end{equation}
where $A_\pm\doteq(2d_2^\pm+d_1^{\pm 2}+\mu_\pm d_2^{\pm 2})=2(c_\pm^2-1)/\mu_\pm$.  Computing the traces in theorem \ref{thm:exp} and equation \eqref{eq:2}, we obtain ${\rm tr}\Lambda=4+d_2^+{\rm tr}L_+^2+d_2^-{\rm tr}L_-^2$ and ${\rm tr}\Lambda^2=4+A_+{\rm tr}L_+^2+A_-{\rm tr}L_-^2$.  Using ${\rm tr}L_\pm^2=-2{\rm tr}_2L_\pm=2\mu_\pm$, one then computes that ($\star$) ${\rm tr}\Lambda=2(c_++c_-)$ and ${\rm tr}\Lambda^2=4(c_+^2+c_-^2-1)$.  Consequently, $(\star\star$) ${\rm tr}_2\Lambda=\tfrac{1}{2}({\rm tr}^2\Lambda-{\rm tr}\Lambda^2)=4c_+c_-+2$.  Equations ($\star$) and ($\star\star$) can then be solved for $c_\pm$ to obtain the formulas in the statement of the lemma.
\end{proof}

\begin{lemma}\label{lem:simplecriterion}
$\Lambda$ is simple if and only if ${\rm tr}_2\Lambda=2{\rm tr}\Lambda-2$ and ${\rm tr}\Lambda\geq 0$.
\end{lemma}

\begin{proof}
If $\Lambda$ is simple, then $\Lambda=\exp(L)$ for some simple Lorentz bivector $L$.  From theorem \ref{thm:expsimple} and lemma \ref{lem:pow}, ($\ast$) $\tau_1=4-2d_2l_2$.  Moreover, we have $\Lambda^2=I+2d_1L+(2d_2+d_1^2)L^2+2d_1d_2L^3+d_2L^4$, and one computes $\tau_2=\tfrac{1}{2}(\tau_1^2-{\rm tr}\Lambda^2)=6+(d_1^2-6d_2)l_2+d_2^2l_2^2$ using lemma \ref{lem:pow} and $\det{L}=0$.  Thus ($\ast\ast$) $\tau_2-2\tau_1+2=(d_1^2-2d_2)l_2+d_2^2l_2^2$.  By theorem \ref{thm:expsimple}, if $l_2<0$, then $d_2l_2=1-\cosh{x}$ and $d_1^2l_2=-\sinh^2{x}=1-\cosh^2{x}$.  Thus ($\ast$) yields $\tau_1=2(1+\cosh{x})\geq 4$, and ($\ast\ast$) yields $\tau_2-2\tau_1+2=0$.  Similarly, if $l_2>0$, then $\tau_1=2(1+\cos{y})\geq 0$ and $\tau_2-2\tau_1+2=0$; and if $l_2=0$, then $\tau_1=4$ and $\tau_2-2\tau_1+2=0$.

Now suppose that $\Lambda$ is not simple, and write $\Lambda=\exp(L)$ as in lemma \ref{lem:lorentztraces}.  Thus $\tau_2-2\tau_1+2=-4(c_+-1)(1-c_-)$, which is only zero when $c_+=1$ or $c_-=1$.  As $\Lambda$ is nonsimple, $\mu_+>0$ and $\mu_-<0$; whence $c_+>1$, and $c_-=1$ only if $\sqrt{-\mu_-}$ is a nonzero multiple of $2\pi$.  However in the latter case, theorem \ref{thm:exp} implies that $\exp(L)=\exp(L_+)$, which cannot be the case, since $\Lambda$ is assumed to be nonsimple.
\end{proof}

We note that the proof of lemma \ref{lem:simplecriterion} implies that if $\Lambda$ is nonsimple, then $c_+>1$ and $-1\leq c_-<1$.  In particular, ${\rm tr}\Lambda>0$.

\subsection{Simple logarithm}

\begin{theorem}\label{thm:simplelog}
Suppose $\Lambda\in SO^+(g)$ is simple.  For ${\rm tr}\Lambda>0$, let us define $L_\Lambda=\tfrac{1}{2}k(\Lambda-\Lambda^{-1})$, where
\begin{description}
  \item{(i)} $\displaystyle k=\frac{\sqrt{\mu}}{\sinh\sqrt{\mu}}$ and $\sqrt{\mu}=\cosh^{-1}(\tfrac{1}{2}{\rm tr}\Lambda-1)$, if ${\rm tr}\Lambda>4$
  \item{(ii)} $\displaystyle k=\frac{\sqrt{-\mu}}{\sin\sqrt{-\mu}}$ and $\sqrt{-\mu}=\cos^{-1}(\tfrac{1}{2}{\rm tr}\Lambda-1)$, if $0<{\rm tr}\Lambda<4$
  \item{(iii)} $k=1$ and $\mu=0$, if ${\rm tr}\Lambda=4$
\end{description}
Then $L_\Lambda$ is a simple Lorentz bivector with $\exp(L_\Lambda)=\Lambda$ and ${\rm tr}_2L_\Lambda=-\mu$.  For ${\rm tr}\Lambda=0$, $P_\Lambda=\tfrac{1}{2}(I-\Lambda)$ is orthogonal projection onto a nondegenerate two--plane such that $-\pi^2P_\Lambda$ is the square of a simple Lorentz bivector $L$ with $\exp{L}=\Lambda$ and ${\rm tr}_2L=\pi^2$.
\end{theorem}

Thus if ${\rm tr}\Lambda\neq 0$, then $L_\Lambda$ may be taken as the logarithm of $\Lambda$: $L_\Lambda=\log\Lambda$.  In the case ${\rm tr}\Lambda=0$, we may reconstruct $\log\Lambda$ from the two independent four--vectors in the image of $P_\Lambda$; it should be noted that in this case, $\Lambda$ is necessarily an involution: $\Lambda^2=I$.  In general, the logarithm is not unique, see \cite{Shaw}.  In particular if $L$ is simple with ${\rm tr}_2L=2n\pi$, for any integer $n$, then $\Lambda\doteq\exp(L)=I$ by theorem \ref{thm:expsimple}.  In this case, the above theorem gives $L_\Lambda=0$.  A version of the formula for $\log\Lambda$ in the case ${\rm tr}\Lambda\neq 0$ that involves algebra over the complex numbers appears in \cite{Coll}.

\begin{proof}
Write $\Lambda=\exp(L)$ for some simple $L$.  By theorem \ref{thm:expsimple}, we have ($\star$) $d_1L=\tfrac{1}{2}(\Lambda-\Lambda^{-1})$.  As in the proof of lemma \ref{lem:simplecriterion}, ($\star\star$) $\tau_1=4-2d_2l_2$.  

From theorem \ref{thm:expsimple}, we see that $d_1=0$ only if $l_2>0$ and $\sqrt{l_2}$ is a nonzero multiple of $\pi$; by the periodicity of the sine and cosine, we may assume that $l_2=\pi^2$.  In this case, $d_2=2/\pi^2$ and by ($\star\star$), $\tau_1=0$.  Conversely, if $\tau_1=0$, then ($\star\star$) implies $d_2l_2=2$, which only happens if $l_2>0$ and $\sqrt{l_2}$ is a nonzero multiple of $\pi$.  Thus if $\tau_1=0$, then $\Lambda=I+(2/\pi^2)L^2=I-2P_L$, with $P_L$ as in lemma \ref{lem:2plane}. Thus, $P_L=P_\Lambda$.

If $\tau_1\neq 0$, then $d_1\neq 0$; and from ($\star$) and theorem \ref{thm:expsimple}, we only need to deduce the value of $l_2$ in terms of $\Lambda$ in order to obtain the formulas (i)---(iii).  If we assume that $l_2<0$, then theorem \ref{thm:expsimple} states that $d_2l_2=1-\cosh{x}$; and so ($\star\star$) can be solved to yield $\cosh{x}=\tfrac{1}{2}\tau_1-1$.  Note that this equation has a solution for $x>0$ if and only if $\tau_1>4$.  The cases when $l_2>0$ and $l_2=0$ are handled similarly.
\end{proof}

\subsection{Decomposition into simple factors}\label{ssec:factordecomp}

In the case when $\Lambda$ is nonsimple, then $L=\log\Lambda$ is a nonsimple bivector, and we have the orthogonal decomposition $L=L_++L_-$ into simple mutually annihilating summands.  Thus $\Lambda=\exp(L_+)\exp(L_-)$ is a product of commuting simple factors.  We give explicit formulas.

\begin{theorem}\label{thm:lorentzproj}
If $\Lambda\in SO^+(g)$ is nonsimple, then $\Lambda=\exp(L)$ where $L$ is a nonsimple Lorentz bivector with projection operators
$$P_\pm=\pm\frac{\tfrac{1}{2}(\Lambda+\Lambda^{-1})-c_\mp I}{c_+-c_-}$$
where $c_\pm$ are as in lemma \ref{lem:lorentztraces}.
\end{theorem}

\begin{proof}
From theorem \ref{thm:exp}, we have ($\ast$) $\Lambda^{-1}=\exp(-L)=I-d_1^+L_+-d_1^-L_-+d_2^+L_+^2+d_2^-L_-^2$, and we obtain the equation $\tfrac{1}{2}(\Lambda+\Lambda^{-1})=I+d_2^+L_+^2+d_2^-L_-^2$, which we rewrite as ($\circ$) $d_2^+L_+^2+d_2^-L_-^2=-I+\tfrac{1}{2}(\Lambda+\Lambda^{-1})$.  In addition, from equation \eqref{eq:2}, we have $\tfrac{1}{2}(\Lambda^2+\Lambda^{-2})=I+A_+L_+^2+A_-L_-^2$.  However, by lemma \ref{lem:spow}, this equation may be rewritten as ($\circ\circ$) $A_+L_+^2+A_-L_-^2=-\tfrac{1}{2}(\tau_2+2)I+\tfrac{1}{2}\tau_1(\Lambda+\Lambda^{-1})$.  Equations ($\circ$) and ($\circ\circ$) form a linear system in $L_\pm^2$, whose determinant is computed to be $2(c_+-1)(1-c_-)(c_+-c_-)/\mu_+\mu_-$, which is never zero.  After inverting the system ($\circ$), ($\circ\circ$) , one computes
$$L^2=L_+^2+L_-^2
   = \frac{1}{c_+-c_-}(\mu_-c_+-\mu_+c_-)I
       +\tfrac{1}{2}(\mu_+-\mu_-)(\Lambda+\Lambda^{-1})
$$
Equation \eqref{eq:proj} then yields the desired formulas for $P_\pm$.
\end{proof}

\begin{theorem}
Suppose $\Lambda\in SO^+(g)$ is nonsimple.  Let $c_\pm$, and $\mu_\pm$ be defined as in lemma \ref{lem:lorentztraces}.  If $\mu_-\neq -\pi^2$, then $\Lambda=\exp(L)$ with $L$ a nonsimple Lorentz bivector whose orthogonal decomposition $L=L_++L_-$ is given by
$$L_\pm=\mp\frac{1}{(c_+-c_-)s_\pm}
        \left\{\tfrac{1}{2}c_\mp(\Lambda-\Lambda^{-1})
              -\tfrac{1}{4}(\Lambda^2-\Lambda^{-2})\right\}
$$
where $s_+\doteq\sinh\sqrt{\mu_+}/\sqrt{\mu_+}$ and $s_-\doteq\sin\sqrt{-\mu_-}/\sqrt{-\mu_-}$.  Moreover, ${\rm tr}_2L_\pm=-\mu_\pm$.
\end{theorem}

\begin{proof}
From equation \eqref{eq:2} and equation ($\ast$) in the proof of the previous theorem, we obtain the equations $\tfrac{1}{2}(\Lambda-\Lambda^{-1})=d_1^+L_++d_1^-L_-$ and $\tfrac{1}{4}(\Lambda^2-\Lambda^{-2})=d_1^+(1+\mu_+d_2^+)L_++d_1^-(1+\mu_-d_2^-)L_-$.  These define a linear system for $L_\pm$.  The determinant is $-(c_+-c_-)s_+s_-$, which is zero only when $\sqrt{-\mu_-}$ is a nonzero multiple of $\pi$.  Inverting the linear system yields the formulas in the statement of the theorem.
\end{proof}

\begin{theorem}
For nonsimple $\Lambda\in SO^+(g)$, $\Lambda=\Lambda_+\Lambda_-$, where $\Lambda_\pm$ are the commuting simple Lorentz transformations given by
$$\Lambda_\pm=\pm\frac{1}{2(c_+-c_-)}
    \left\{(1+2c_\pm)I-\Lambda^{-1}-(1+2c_\mp)\Lambda+\Lambda^2\right\}
$$
where $c_\pm$ are as in lemma \ref{lem:lorentztraces}.
\end{theorem}

\begin{proof}
Let $L$ be a nonsimple Lorentz bivector such that $\Lambda=\exp(L)$, and $P_\pm$ as in theorem \ref{thm:lorentzproj}.  Observe that $\exp(P_\pm L)=P_\mp+P_\pm\Lambda$.  Indeed, $(P_\pm L)^n=P_\pm^nL^n=P_\pm L^n$ for all $n>0$, since $P_\pm$ is a projection that commutes with $L$.   Thus, $\exp(P_\pm L)=I+\sum_{n>0}\frac{1}{n!}(P_\pm L)^n=I+P_\pm(\sum_{n>0}\frac{1}{n!}L^n)=I+P_\pm(\exp(L)-I)$.  The formulas for $\Lambda_\pm$ then follow from theorem \ref{thm:lorentzproj}.
\end{proof}

As we noted in section \ref{sssec:spin}, in the special case when $g=\eta$ is the Minkowski metric, we may write $L_\pm$ in terms of the Lie algebra homomorphism $\psi:{\it sl_2}{\mathbb C}\rightarrow{\it so}(\eta)$.  Namely, $\psi(m_+)=L_+$ and $\psi(i\theta m_+)=L_-$ for some traceless $2\times 2$ complex matrix $m_+=\bigl(\begin{smallmatrix} a & b\\ c & -a\end{smallmatrix}\bigr)$ such that $a^2+bc$ is real.  It follows that $M_+\doteq\exp(m_+)$ and $M_-\doteq\exp(i\theta m_+)$ give (commuting) matrices in ${\it SL_2}{\mathbb C}$ with $\Psi(M_\pm)=\Lambda_\pm$; i.e., $M_\pm$ gives the decomposition of $\Psi(M_+M_-)=\Lambda$ into simple factors.  Observe that $\exp(m_+)$ is readily computed, since $m_+^2=(a^2+bc)I$.

We may view the decomposition of a nonsimple Lorentz transformation into commuting simple factors as a generalization of Synge's {\it 4--screw:} a product of a boost and a rotation in orthogonal two--planes.  While it is known that every proper nonsimple Lorentz transformation can be expressed as the product of two commuting simple Lorentz transformations (see \cite{Schremp}), the formulas presented here are ostensibly original, and in any case independent of the specific form of the Lorentz metric.


\end{document}